\documentclass[10pt]{article}

\usepackage{amssymb,amsmath}
\newtheorem{thm}{Theorem}[section]

\newtheorem{lem}[thm]{Lemma}

\newenvironment{proof}{\begin{trivlist}
                       \item[]{\bf Proof.}
                       \hspace{0cm}}{\hfill $\Box$
                       \end{trivlist}}

 \newcommand{\thmref}[1]{Theorem~\ref{#1}}
 \newcommand{\lemref}[1]{Lemma~\ref{#1}}

\newcommand{\R}{{\mathbb R}}
\newcommand{\C}{{\mathbb C}}

\newcommand{\bee}{\begin{equation*}}
\newcommand{\eee}{\end{equation*}}
\newcommand{\be}{\begin{equation}}
\newcommand{\ee}{\end{equation}}
\newcommand{\pn}{\par\noindent}

\title{Uniqueness of the solution to inverse scattering problem with
scattering data at a fixed direction of the incident wave}
\author{A G Ramm\\
\small Department of Mathematics\\[-0.8ex]
\small Kansas State University, Manhattan, KS 66506-2602, USA\\[-0.8ex]
\small \texttt{ramm@math.ksu.edu}\\
}

\begin{document}

\date{}
\maketitle
\begin{abstract}
Let $q(x)$ be real-valued compactly supported 
sufficiently smooth function. It is proved that the scattering data 
$A(\beta,\alpha_0,k)$
$\forall \beta\in S^2$, $\forall k>0,$ determine $q$ uniquely.
Here $\alpha_0\in S^2$ is a fixed direction of the incident plane wave.
\end{abstract}
\pn{\\ MSC: 35P25, 35R30, 81Q05; PACS: 02.30.Jr; 02.30.Zz; 03.65.-w  \\
{\em Key words:} inverse scattering; non-overdetermined data; 
fixed direction of the incident wave scattering data. }

\section{Introduction}
The scattering solution $u(x,\alpha,k)$
solves the scattering problem: \be\label{e1}
[\nabla^2+k^2-q(x)]u=0\quad in\quad \R^3,\ee \be\label{e2}
u=e^{ik\alpha\cdot
x}+A(\beta,\alpha,k)\frac{e^{ikr}}{r}+o\left(\frac{1}{r}\right),\quad
r:=|x|\to \infty,\ \beta:=\frac{x}{r}. \ee 
Here $\alpha, \beta \in S^2$ are the unit vectors, $S^2$ is the unit 
sphere, the coefficient
$A(\beta,\alpha,k)$ is called the scattering amplitude, $q(x)$
is a real-valued compactly supported sufficiently smooth function.  
The inverse scattering problem
of interest is to determine $q(x)$ given the scattering
data $A(\beta, \alpha_0, k)$ $\forall \beta\in S^2, \, \forall k>0$. 
This problem is called {\it the inverse scattering problem with
fixed direction of the incident plane wave data.} 

The function
$A(\beta,\alpha_0,k)$ depends on one unit vector $\beta$
and on the scalar $k$, i.e., on three variables.
The potential $q(x)$ depends also on three variables $x\in \R^3$.
This inverse problem is, therefore, not over-determined
in the sense that the data and the unknown $q(x)$ are functions
of the same number of variables. 

{\it Historical remark.} In the beginning of the forties of the last
century physicists raised the the following question: is it possible to
recover the Hamiltonian of a quantum-mechanical system from the
observed quantities, such as $S$-matrix? In the non-relativistic
quantum mechanics  the simplest Hamiltonian ${\bf H}=-\nabla^2+q(x)$
can be uniquely determined
if one knows the potential $q(x)$. The $S$-matrix in this case is
in one-to-one correspondence with the scattering amplitude $A$:
$S=I-\frac k{2\pi i}A$, where $I$ is the identity operator in
$L^2(S^2)$,  $A$ is an integral operator in $L^2(S^2)$ with
the kernel $A(\beta, \alpha, k)$, and $k^2>0$ is energy.
Therefore, the question, raised by the
physicists, is reduced to an inverse scattering problem: can one determine
the potential $q(x)$ from the knowledge of the scattering amplitude.
The inverse scattering problem with fixed direction $\alpha_0$ of the 
incident plane wave scattering data  $A(\beta, \alpha_0, k)$,
known for all $\beta\in S^2$ and all $k>0$ has been open from
the forties of the last century. In this paper we prove uniqueness of
the solution to this inverse problem under the {\it Assumption A)}
formulated below. Altough there is a large literature on inverse
scattering (see, e.g., references in \cite{R470}, \cite{Ka}), the above 
problem was
not solved, and the references we give are only to the papers directly 
related to our presentation.   

Let $B_a$ be the ball centered at the
origin and of radius $a$, and $H_0^\ell(B_a)$ be the closure of
$C_0^\infty(B_a)$ in the norm of the Sobolev space $H^\ell(B_a)$ of
functions whose derivatives up to the order $\ell$ belong to
$L^2(B_a).$

{\it Assumption A): 

We assume that $q$ is compactly supported, i.e.,
$q(x)=0$ for $|x|>a$, where $a>0$ is an arbitrary large fixed
number; $q(x)$ is real-valued, i.e., $q=\overline{q}$; and $q(x)\in
H_0^\ell(B_a)$, $\ell>3$.}

It was proved in  \cite{R35}
(see also \cite{R486}, Chapter 6), 
that if $q=\overline{q}$ and $q\in L^2(B_a)$
is compactly supported, then the resolvent kernel $G(x,y,k)$ of the
Schr\"{o}dinger operator $-\nabla^2+q(x)-k^2$ is a meromorphic
function of $k$ on the whole complex plane $k$, analytic in Im$k\geq
0$, except, possibly, of a finitely many simple poles at the points
$ik_j$, $k_j>0$, $1\leq j\leq n$, where $-k_j^2$ are negative
eigenvalues of the selfadjoint operator $-\nabla^2+q(x)$ in
$L^2(\R^3)$. Consequently, the scattering amplitude
$A(\beta,\alpha,k)$, corresponding to the above $q$, is a restriction
to the positive semiaxis $k\in[0,\infty)$ of a meromorphic on the
whole complex $k$-plane function.

It was proved by
the author (\cite{R228}), that the {\it fixed-energy scattering data}
$A(\beta,\alpha):=A(\beta,\alpha,k_0)$, $k_0=const>0$, $\forall
\beta \in S_1^2,$ $\forall \alpha\in S_2^2$, determine real-valued
compactly supported $q\in L^2(B_a)$ uniquely. Here $S_j^2,$ $j=1,2$,
are arbitrary small open subsets of $S^2$ (solid angles).
No uniqueness results for the potentials which decay at a
power rate are known if the scattering data are known at a fixed
energy. If the potentials decay faster than exponentially as $|x|\to 
\infty$,
then a uniqueness result for this problem is obtained in \cite{RS}. 
If the potential decays at a power rate but the scattering 
data are known for all $k>0$, all $\alpha\in S^2$ and all $\beta\in S^2$
then a uniqueness results was obtained in \cite{Sa}.
 
In \cite{R425} (see also monograph \cite{R470}, Chapter 5, and 
\cite{R285}) an analytical 
formula is derived for the reconstruction of the potential 
$q$ from  exact fixed-energy scattering data, and from  noisy
fixed-energy scattering data, and stability estimates and error 
estimates for the reconstruction method are obtained. 
To the author's knowledge, these are
the only known until now  theoretical error estimates for the recovery 
of the potential from noisy fixed-energy scattering data in
the three-dimensional inverse scattering problem.

In \cite{R325} stability results are obtained for the inverse 
scattering problem for obstacles.

The scattering data $A(\beta,\alpha)$ depend on four variables (two
unit vectors), while the unknown $q(x)$ depends on three variables.
In this sense the inverse scattering problem, which consists of finding 
$q$ from the fixed-energy scattering data $A(\beta,\alpha)$, is
overdetermined.

The first uniqueness theorem for three-dimensional inverse scattering  
problem with non-overdetermined data was proved by the author 
in  \cite{R584}, where the scattering data were the backscattering 
data $A(-\beta, \beta, k)$ $\forall \beta\in S^2$ $\forall k>0$.
The goal of this paper is to prove  a uniqueness theorem
for the three-dimensional inverse scattering problem with
the scattering data  $A(\beta, \alpha_0, k)$ $\forall \beta\in S^2$ 
$\forall k>0$. These data are also non-overdetermined.
Our work is based on the method developed in \cite{R589}, but the
presentation is self-contained.

\begin{thm}\label{thm1}
If Assumption A) holds,  then the data
$A(\beta, \alpha_0,k)$ $\forall \beta\in S^2$, $\forall k>0$, 
and a fixed $\alpha_0\in S^2$, determine
$q$ uniquely.
\end{thm}
{\bf Remark 1.} The conclusion of \thmref{thm1} remains valid
if the data $A(\beta, \alpha_0,k)$ are known $\forall \beta\in S_1^2$
and $k\in(k_0,k_1),$ where $(k_0,k_1)\subset[0,\infty)$ is an
arbitrary small interval,  $k_1>k_0$, and $S_1^2$ is an arbitrary
small open subset of $S^2.$

In Section 2 we formulate some 
auxiliary results. 

In Section 3 proof of \thmref{thm1} is given.

In the Appendix a technical estimate is proved.

\section{Auxiliary results}
Let \be\label{e3} F(g):=\tilde{g}(\xi)=\int_{\R^3}g(x)e^{i\xi\cdot
x} dx,\quad g(x)=\frac{1}{(2\pi)^3}\int_{\R^3} e^{-i\xi\cdot
x}\tilde{g}(\xi) d\xi. \ee

If $f*g:=\int_{\R^3}f(x-y)g(y)dy,$ then \be\label{e4}
F(f*g)=\tilde{f}(\xi)\tilde{g}(\xi),\quad
F(f(x)g(x))=\frac{1}{(2\pi)^3}\tilde{f}*\tilde{g}. \ee If
\be\label{e5} G(x-y,k):=\frac{e^{ik[|x-y|-\beta\cdot(x-y)]}
}{4\pi|x-y|},\ee 
then 
\be\label{e6}
F(G(x,k))=\frac{1}{\xi^2-2k\beta\cdot \xi}, \qquad \xi^2:=\xi\cdot
\xi. 
\ee 
The scattering solution $u=u(x,\alpha,k)$ solves (uniquely)
the integral equation \be\label{e7} u(x,\alpha,k)=e^{ik\alpha\cdot
x}-\int_{B_a}g(x,y,k)q(y)u(y,\alpha,k)dy, \ee where \be\label{e8}
g(x,y,k):=\frac{e^{ik|x-y|}}{4\pi|x-y|}.\ee If \be\label{e9}
v=e^{-ik\alpha\cdot x}u(x,\alpha,k),\ee then \be\label{e10}
v=1-\int_{B_a}G(x-y,k)q(y)v(y,\alpha,k)dy,\ee where $G$ is defined
in \eqref{e5}.

Define $\epsilon$ by the formula \be\label{e11} v=1+\epsilon.\ee
Then \eqref{e10} can be rewritten as 
\be\label{e12}
\epsilon(x,\alpha,k)=-\int_{\R^3}G(x-y,k)q(y)dy- T\epsilon,
\ee
where 
$$T\epsilon:=\int_{B_a}G(x-y,k)q(y)\epsilon(y,\alpha,k)dy.$$
Fourier transform of \eqref{e12} yields (see \eqref{e4},\eqref{e6}):
\be\label{e13}
\tilde{\epsilon}(\xi,\alpha,k)=-\frac{\tilde{q}(\xi)}{\xi^2-2k\alpha\cdot
\xi}-\frac{1}{(2\pi)^3}\frac{1}{\xi^2-2k\alpha\cdot
\xi}\tilde{q}*\tilde{\epsilon}. \ee 
An essential ingredient of our
proof in Section 3 is the following lemma, proved by the author
 in  \cite{R470}, p.262, and in  \cite{R425}. For convenience of the 
reader
a short proof of this lemma is given in Appendix.

\begin{lem}\label{lem1}
If $A_j(\beta,\alpha,k)$ is the scattering amplitude corresponding
to potential $q_j$, $j=1,2$, then \be\label{e14}
-4\pi[A_1(\beta,\alpha,k)-A_2(\beta,\alpha,k)]=\int_{B_1}[q_1(x)-q_2(x)]
u_1(x,\alpha,k)u_2(x,-\beta,k)dx,
\ee where $u_j$ is the scattering solution corresponding to $q_j$.
\end{lem}
Consider an algebraic variety $\mathcal{M}$ in $\C^3$ defined by the
equation 
\be\label{e15} \mathcal{M}:=\{\theta\cdot\theta=1,\quad
\theta\cdot\theta:=\theta_1^2+\theta_2^2+\theta_3^2,\quad \theta_j\in
\C,\,\, 1\leq j \leq 3.\} \ee This is a non-compact variety, intersecting 
$\R^3$ over the
unit sphere $S^2$.

Let $R_+=[0,\infty).$ The following result is proved in 
\cite{R190}, p.62.
\begin{lem}\label{lem2}
If Assumption A) holds, then the scattering amplitude
$A(\beta,\alpha,k)$ is a restriction to $S^2\times S^2\times R_+$ of
a function $A(\theta',\theta,k)$ on
$\mathcal{M}\times\mathcal{M}\times\C$, analytic on
$\mathcal{M}\times\mathcal{M}$ and meromorphic on $\C$, $\theta'$,
$\theta\in \mathcal{M}$, $k\in \C$.
\end{lem}
The scattering solution $u(x,\alpha,k)$ is a meromorphic function of
$k$ in $\C$, analytic in Im$k\geq 0$, except, possibly, at the points
$k=ik_j$, $1\leq j\leq n$, $k_j>0$, where $-k_j^2$ are negative 
eigenvalues of the selfadjoint Schr\"odinger operator, defined by the 
potential $q$ in $L^2(\R^3)$. These eigenvalues can be absent, for 
example, if $q\geq 0$.

We need the notion of the Radon transform: 
\be\label{e16}
\hat{f}(\beta,\lambda):=\int_{\beta\cdot x=\lambda}f(x)d\sigma, \ee
where $d\sigma$ is the element of the area of the plane $\beta\cdot
x=\lambda$, $\beta\in S^2$, $\lambda$ is a real number. 
The following properties of the Radon transfor will be used:
\be\label{e17}
\int_{B_a}f(x)dx=\int_{-a}^a\hat{f}(\beta,\lambda)d\lambda, \ee
\be\label{e18} \int_{B_a}e^{ik\beta\cdot x}f(x)dx=\int_{-a}^a
e^{ik\lambda}\hat{f}(\beta,\lambda) d\lambda, \ee \be\label{e19}
\hat{f}(\beta,\lambda)=\hat{f}(-\beta,-\lambda). \ee 
These properties are proved, e.g., in \cite{R348}, pp. 12, 15.

We also need the following  Phragmen-Lindel\"{o}f lemma, which is proved 
in  \cite{Le}, p.69,  and  in \cite{PS}.
\begin{lem}\label{lem3}
Let $f(z)$ be holomorphic inside an angle $\mathcal{A}$ of opening
$<\pi$;  $|f(z)|\leq c_1e^{c_2|z|}, \quad z\in \mathcal{A},$
$c_1,c_2>0$ are constants; $|f(z)|\leq
M$ on the boundary of $\mathcal{A}$; and $f$ is continuous up to the 
boundary
of $\mathcal{A}$. Then $|f(z)|\leq M,\quad \forall z\in \mathcal{A}.$
\end{lem}

\section{Proof of \thmref{thm1}}
The scattering data in Remark 1
determine uniquely the scattering data in \thmref{thm1} by
\lemref{lem2}.

{\it Let us outline the ideas of the proof of \thmref{thm1}.}

Assume that potentials
$q_j$, $j=1,2$, generate the same scattering data:
$$A_1(\beta, \alpha_0,k)=A_2(\beta, \alpha_0,k)\qquad \forall \beta\in 
S^2,
\quad \forall k>0,$$ and let
$$p(x):=q_1(x)-q_2(x).$$
Then by \lemref{lem1},
see equation \eqref{e14}, one gets \be\label{e20}
0=\int_{B_a}p(x)u_1(x,\alpha_0,k)u_2(x,-\beta,k)dx, \qquad \forall 
\beta\in
S^2,\ \forall k>0.\ee 
By \eqref{e9} and \eqref{e11} one can rewrite \eqref{e20} as 
\be\label{e21} \int_{B_a}e^{i\kappa \zeta\cdot
x}[1+\epsilon(x,k)]p(x)dx=0\quad \forall \zeta\in S^2_+,\ \forall k>0,
\ee where \bee
\epsilon(x,k):=\epsilon:=\epsilon_1(x,k)+\epsilon_2(x,k)+
\epsilon_1(x,k)\epsilon_2(x,k),
\eee 
and we have denoted $|\alpha_0-\beta|:=\tau,$ 
$\zeta:=(\alpha_0-\beta)/\tau$, $\kappa:=\tau k$.
Without loss of generality one may assume that $\alpha_0$ is the unit
vector along $x_3-$ axis. Then $\tau$ runs through $[0,2]$ and 
the unit vector $\zeta$
 runs through  $S^2_+$, the upper half  of the unit sphere $S^2$. 
Since $k\in [0, \infty)$ is arbitrary in \eqref{e21}, so is $\kappa=\tau k$.
Because the left-hand side of \eqref{e21} depends on $\zeta$ analytically
on the variety $\mathcal{M}$, one concludes that 
relation \eqref{e21} holds for any $\zeta\in S^2$ if it holds for
$\zeta\in S^2_+$. So, from now on we will use formula   \eqref{e21}
with $\zeta \in S^2$ being arbitrary. 

By \lemref{lem2} the relations \eqref{e20} and \eqref{e21} hold for 
complex $k$, 
\be\label{e22} \tau k=\kappa+i\eta,
\quad \eta\geq 0. 
\ee
Using formulas \eqref{e3}-\eqref{e4}, one derives
from \eqref{e21} the relation 
\be\label{e23}
\tilde{p}((\kappa+i\eta)\zeta)+\frac{1}{(2\pi)^3}
(\tilde{\epsilon}*\tilde{p})((\kappa+i\eta)\zeta)=0
\qquad \forall \zeta\in S^2,\, \forall \kappa\in \R, \ee
where the notation $(f*g)(z)$ means that the convolution $f*g$ is 
calculated at the argument $z=(\kappa+i\eta)\zeta$.
  
One has 
\be\label{e24} \sup_{\zeta\in
S^2}|\tilde{\epsilon}*\tilde{p}|:= \sup_{\zeta\in
S^2}|\int_{\R^3}\tilde{\epsilon}((\kappa+i\eta)\zeta-s)\tilde{p}(s)ds|\leq
\nu(\kappa, \eta)\sup_{s\in \R^3}|\tilde{p}(s)|, \ee 
where 
$$\nu(\kappa, \eta):=\sup_{\zeta\in S^2}\int_{\R^3}|\tilde{\epsilon}
((\kappa+i\eta)\zeta-s)|ds.$$
We  prove that if  $\eta=\eta(\kappa)=O(\ln \kappa)$
is suitably chosen, namely as in \eqref{e29} below,  
then the following inequality holds:
\be\label{e25}
0<\nu(\kappa,\eta(\kappa))<1, \qquad \kappa\to \infty. \ee 
We also prove that
\be\label{e26} \sup_{\zeta\in
S^2}|\tilde{p}((\kappa+i\eta(\kappa))\zeta)|\geq \sup_{s\in
\R^3}|\tilde{p}(s)|,\quad \kappa\to \infty, \ee 
and then it follows from
\eqref{e23}-\eqref{e26} that $\tilde{p}(s)=0$, so $p(x)=0$, and
\thmref{thm1} is proved. 
Indeed, it follows from \eqref{e23} and \eqref{e26} that,
for sufficiently large $\kappa$ and a suitable $\eta(k)=O(\ln k)$, 
one has 
$$ \sup_{s\in \R^3}|\tilde{p}(s)|\leq \frac 1 {(2\pi)^3} 
\nu(\kappa,\eta(\kappa))\sup_{s\in \R^3}|\tilde{p}(s)|.$$
If \eqref{e25} holds, then the above equation implies that $\tilde{p}=0$.
This and the injectivity of the Fourier transform imply that $p=0$. 

{\it This completes the outline of the proof of \thmref{thm1}}.

Let us now give a detailed proof of estimates \eqref{e25} 
and \eqref{e26}, that completes the proof of Theorem 1.1. 
We denote $\zeta$ by $\beta$ in what follows, since both
unit vectors run through all of $S^2$.
 
We assume that $p(x)\not\equiv 0$, because
otherwise there is nothing to prove. Let
$$\max_{s\in\R^3}|\tilde{p}(s)|:=\mathcal{P}\neq 0.$$

\begin{lem}\label{lem4}
If Assumption A) holds and $\mathcal{P}\neq 0$, then
\be\label{e27} \limsup_{\eta\to \infty}
\max_{\beta\in S^2}|\tilde{p}((\kappa+i\eta)\beta)|=\infty, \ee
where $\kappa>0$ is arbitrary but fixed.
For any $\kappa>0$ there is an $\eta=\eta(\kappa)$, such that
\be\label{e28}
\max_{\beta\in S^2}|\tilde{p}((\kappa+i\eta(\kappa))\beta)|=
\mathcal{P}, 
\ee
where the number $\mathcal{P}:=\max_{s\in\R^3}|\tilde{p}(s)|$, 
and
\be\label{e29} \eta(\kappa)=a^{-1}\ln \kappa +O(1)  \text{\quad as 
\quad }\kappa\to
+\infty. \ee
\end{lem}

{\it Proof of \lemref{lem4}}.    By formula \eqref{e18} one
gets \be\label{e30}
\tilde{p}((\kappa+i\eta)\beta)=\int_{B_a}p(x)e^{i(\kappa+i\eta)\beta\cdot
x}dx=\int_{-a}^ae^{i\kappa \lambda
-\eta\lambda}\hat{p}(\beta,\lambda)d\lambda. \ee 
The function
$\hat{p}(\beta,\lambda)$ is compactly supported, real-valued, and 
satisfies relation \eqref{e19}. Therefore
\be\label{e31} \max_{\beta\in
S^2}|\tilde{p}((\kappa+i\eta(\kappa))\beta)|=\max_{\beta\in
S^2}|\tilde{p}((\kappa-i\eta(\kappa))\beta)|. \ee 
Indeed,
\be\begin{split}\label{e32} \max_{\beta\in
S^2}|\tilde{p}((\kappa+i\eta(\kappa))\beta)|&=\max_{\beta\in
S^2}\left|\int_{-a}^ae^{i\kappa\lambda-\eta\lambda}\hat{p}(\beta,\lambda)
d\lambda\right|\\
&=\max_{\beta\in
S^2}\left|\int_{-a}^ae^{-i\kappa \mu+\eta\mu}\hat{p}(\beta,-\mu)d\mu\right|\\
&=\max_{\beta'\in S^2}\left| \int_{-a}^a e^{-i\kappa
\mu+\eta\mu}\hat{p}(-\beta',-\mu)d\mu
\right|\\
&=\max_{\beta'\in S^2}\left| \int_{-a}^a e^{-i\kappa
\mu+\eta\mu}\hat{p}(\beta',\mu)d\mu
\right|\\
&=\max_{\beta\in S^2}|\tilde{p}((\kappa-i\eta)\beta)|.\\
\end{split}\ee 
At the last step we took into account that
$\hat{p}(\beta,\lambda)$ is a real-valued function, so
 \begin{equation}\begin{split}\max_{\beta\in S^2}\left| \int_{-a}^a 
e^{-i\kappa \mu+\eta\mu}\hat{p}(\beta,\mu)d\mu
\right|&=\max_{\beta\in S^2}\left| \int_{-a}^a 
e^{i\kappa
\mu+\eta\mu}\hat{p}(\beta,\mu)d\mu\right|\\
&=
\max_{\beta\in S^2}|\tilde{p}((\kappa-i\eta)\beta)|.
\end{split}\end{equation}
If $p(x)\not\equiv 0$, then \eqref{e30} and \eqref{e31}
imply \eqref{e27}, as follows from \lemref{lem3}. Let us give a detailed 
proof of this statement.

Consider
the function $h$ of the complex variable  $z:=\kappa+i\eta:$ 
\be\label{e33}
h:=h(z,\beta):=\int_{-a}^ae^{iz\lambda}\hat{p}(\beta,\lambda)d\lambda.\ee
If \eqref{e27} is false, then \be\label{e34} |h(z,\beta)|\leq c\quad
\forall z=\kappa+i\eta,\quad \eta\geq 0,\quad \forall \beta\in S^2,
\ee where $\kappa\geq 0$ is an arbitrary  fixed number and the constant 
$c>0$ does not depend on $\beta$ and $\eta$. 

Thus, $|h|$ is bounded on the ray $\{\kappa=0, \eta\geq 0\}$, which is 
part of the boundary of the right angle $\mathcal{A}$, and the other
part of its boundary is the ray  $\{\kappa\geq 0, \eta= 0\}$. Let us check 
that $|h|$ is bounded on this ray also.

One has
\be\label{e35} |h(\kappa,\beta)|=
|\int_{-a}^ae^{i\kappa \lambda}\hat{p}(\beta,\lambda)d\lambda|\leq
\int_{-a}^a|\hat{p}(\beta,\lambda)|d\lambda\leq c, \ee 
where $c$ stands in this paper for {\it various} constants. 
From \eqref{e34}-\eqref{e35} it
follows that on the boundary of the right angle $\mathcal{A}$, namely, 
on the two
rays $\{\kappa\geq 0, \eta=0\}$ and $\{\kappa=0, \eta\geq 0\}$ the
entire function $h(z,\beta)$ of the complex variable $z$ is bounded, 
$|h(z,\beta)|\leq c$,
and inside $\mathcal{A}$ this function  satisfies 
the estimate
\be\label{e36} |h(z,\beta)|\leq
e^{a|\eta|}\int_{-a}^a|\hat{p}(\beta,\lambda)|d\lambda\leq
ce^{a|\eta|}, \ee
where $c$ does not depend on $\beta$.
Therefore, by \lemref{lem3}, $|h(z,\beta)|\leq c$ in the
whole angle $\mathcal{A}$. 

By \eqref{e31} the same argument is applicable to the remaining three
right angles, the union of which is the whole complex $z-$plane $\C$. 
Therefore 
\be\label{e37} \sup_{z\in
\C,\beta\in S^2}|h(z,\beta)|\leq c. \ee 
This implies by the Liouville theorem  that
$h(z,\beta)=c \,\, \forall z\in \C.$ 

Since $\hat{p}(\beta,\lambda)\in L^1(-a,a)$, the
relation 
\be\label{e38}
\int_{-a}^ae^{iz\lambda}\hat{p}(\beta,\lambda)d\lambda=c\quad \forall
z\in \C, 
\ee 
and the Riemann-Lebesgue lemma imply that $c=0$, so 
$\hat{p}(\beta,\lambda)=0$ $\forall \beta\in S^2$ and $\forall \lambda\in 
\R$. Therefore
$p(x)=0$, contrary to our assumption. Consequently,  relation
\eqref{e27} is proved. 

Relation \eqref{e28} follows from \eqref{e27} because for large
$\eta$ the left-hand side of \eqref{e28} is larger than $\mathcal{P}$
due to \eqref{e27}, while for $\eta=0$ the left-hand side of \eqref{e28} 
is not larger than  $\mathcal{P}$ by the definition of the Fourier 
transform.  

Let us derive estimate \eqref{e29}.

From the assumption $p(x)\in H^{\ell}_0(B_a)$ it
follows that 
\be\label{e39} |\tilde{p}((\kappa+i\eta)\beta)|\leq
c\frac{e^{a|\eta|}}{(1+\kappa^2+\eta^2)^{\ell/2}}. \ee
This inequality is proved in Lemma 3.2, below.

The right-hand side of this inequality is of the order $O(1)$ as
$\kappa\to
\infty$ if $|\eta|=a^{-1}\ln \kappa+O(1)$ as $\kappa \to 
\infty$. This proves  
relation \eqref{e29} and  we specify $O(\ln \kappa)$ as in this 
relation.

Let us now prove inequality  \eqref{e39}. 
\begin{lem}\label{lem5}
If $p\in H_0^\ell(B_a)$ then estimate \eqref{e39} holds.
\end{lem}
\begin{proof}
Consider $\partial_jp:=\frac{\partial p}{\partial x_j}.$ One has
\be\begin{split}
\left|\int_{B_a}\partial_jpe^{i(\kappa+i\eta)\beta\cdot x}
dx\right|&=\left|-i(\kappa+i\eta)\beta_j\int_{B_a}p(x)
e^{i(\kappa+i\eta)\beta\cdot x} dx\right|\\
&= (\kappa^2+\eta^2)^{1/2}|\tilde{p}((\kappa+i\eta)\beta)|.\\
\end{split}\ee
The left-hand side of the above formula admits the following estimate
$$\left|\int_{B_a}\partial_jpe^{i(\kappa+i\eta)\beta\cdot x}
dx\right|\leq ce^{a|\eta|},$$
where the constant $c>0$ is proportional to $||\partial_jp||_{L^2(B_a)}$. 
Therefore,  
\be\label{e36'} |\tilde{p}((\kappa+i\eta)\beta)|\leq
c[1+(\kappa^2+\eta^2)]^{-1/2}e^{a|\eta|}.
\ee
Repeating this argument
one gets estimate \eqref{e39}. Lemma 3.2 is proved.
\end{proof}
Estimate \eqref{e36'} implies that if relation \eqref{e29} holds
and $\kappa \to \infty$, then  
the quantity
$\sup_{\beta \in S^2}|\tilde{p}((\kappa+i\eta)\beta)|$
remains bounded as $\kappa\to \infty$.

If $\eta$ is fixed and $\kappa \to \infty$, then $\sup_{\beta \in
S^2}|\tilde{p}((\kappa+i\eta)\beta)| \to 0$ by the Riemann-Lebesgue
lemma. This, the continuity of $|\tilde{p}((\kappa+i\eta)\beta)|$
 with respect to $\eta$, and  relation
\eqref{e27},  imply the existence of $\eta=\eta(\kappa)$,
such that  equality \eqref{e28} holds, and, consequently, 
inequality  \eqref{e26} holds. 
This  $\eta(\kappa)$ satisfies
 \eqref{e29} because $\mathcal{P}$ is bounded.

Lemma 3.1 is proved \hfill $\Box$

To complete the
proof of Theorem 1.1
one has to establish estimate \eqref{e25}. This
estimate will be established if one proves the following relation:
\be\label{e40} \lim_{\kappa\to \infty}\nu(\kappa):=\lim_{\kappa\to
\infty}\nu(\kappa,\eta(\kappa)) =0,
\ee 
where $\eta(\kappa)$ satisfies \eqref{e29} and 
\be\label{e41} \nu(\kappa,\eta)=\sup_{\beta\in
S^2}\int_{\R^3}|\tilde{\epsilon}((\kappa+i\eta)\beta-s)|ds. \ee

Our argument is valid for $\epsilon_1$, $\epsilon_2$ and
$\epsilon_1\epsilon_2$, so we will use the letter $\epsilon$ and
equation \eqref{e13} for $\tilde{\epsilon}$. 

Below we denote  $2k:=\kappa+i\eta$ and we choose
$\eta=\eta(\kappa)=a^{-1}\ln \kappa +O(1)$ as $\kappa \to \infty$.

We prove that equation \eqref{e12} can be solved by iterations if 
Im$\eta\geq 0$ and $|k+i\eta|$  is 
sufficiently large,  because for such $k+i\eta$ the operator $T^2$ 
 has small norm in $C(B_a)$, the space
of functions, continuous in the ball $B_a$,  with the $\sup$-norm.
Since equation \eqref{e12}  can be solved by iterations and the norm
of $T^2$ is small, the main term 
in the series, representing its solution, as
$|\kappa+i\eta|\to \infty$, $\eta\geq 0$, is the free term of the equation
\eqref{e12}.  The same is true for the Fourier transform
of equation \eqref{e12}, i.e., for equation \eqref{e13}.
Therefore the main term  of the solution $\tilde{\epsilon}$
to equation \eqref{e13} as $|\kappa+i\eta|\to \infty$, $\eta\geq 0$, is 
obtained by using the estimate of the free term of this equation. 
Thus,
it is sufficient to check estimate \eqref{e40} for the function 
$\nu(\kappa, \eta(\kappa))$ using in place of $\tilde{\epsilon}$ the 
function $\tilde{q}(\xi)(\xi^2-2k\beta\cdot\xi)^{-1}$, with $2k$ replaced 
by $\kappa+i\eta$ and $\eta=a^{-1}\ln \kappa +O(1)$ as $\kappa \to 
\infty$.
  
For the above claim that equation \eqref{e12} has the operator
$$T\epsilon=\int_{B_a}G(x-y,k)q(y)\epsilon(y,\beta, k)dy,$$  with the norm
$||T^2||$ in the space $C(B_a)$, which tends to zero as
$|\kappa+i\eta|\to \infty,$ $\eta\geq 0$,  see  Appendix. 

Thus, let us estimate the modulus of the factor 
$\nu(\kappa, \eta)$ in \eqref{e24} with $\eta=\eta(\kappa)$ 
as in \eqref{e29}. 
Using inequality \eqref{e39}, and denoting $\xi=(\kappa+i\eta)\beta$,
where $\beta\in S^2$ plays the role of $\alpha$ in \eqref{e13}, one 
obtains:
\be\label{e42}\begin{split} I:&=\sup_{\beta\in
S^2}\int_{\R^3}\frac{|\tilde{q}((\kappa+i\eta)\beta-s)|
ds}{|[(\kappa+i\eta)\beta-s)^2-(\kappa+i\eta)\beta 
\cdot ((\kappa+i\eta)\beta-s)]|}\\
&\leq c e^{a|\eta|}\sup_{\beta\in S^2}
\int_{\R^3}\frac{ds}{|s^2-(\kappa+i\eta)\beta\cdot
s|[1+(\kappa\beta-s)^2+\eta^2]^{\ell/2}}\\
&:=ce^{a|\eta|}J.\end{split}\ee

Let us prove that 
$$J=o(\frac 1 \kappa),\,\, \kappa\to \infty.$$ 
If this estimate is proved and $\eta=a^{-1}\ln \kappa +O(1)$, then 
$I=o(1)$ as $ \kappa\to \infty,$ therefore
relation \eqref{e40} follows, and Theorem 1.1 is proved.

Let us write the integral $J$ in the spherical coordinates with
$x_3$-axis directed along vector $\beta$. We have
$$|s|=r, \qquad \beta\cdot s=r\cos\theta :=rt,\quad -1\leq t\leq 1.$$
Denote
$$\gamma:=\kappa^2+\eta^2.$$ 
Then
\be\begin{split}
J&\leq 2\pi\int_0^\infty dr r \int_{-1}^1\frac{dt}{[(r-\kappa
t)^2+\eta^2t^2]^{1/2}(1+\gamma+r^2-2r\kappa t)^{\ell/2}}\\
&:=2\pi \int_0^\infty dr r B(r),
\end{split}\ee
where 
$$B:=B(r)=B(r,\kappa,\eta):=\int_{-1}^1\frac{dt}{[(r-\kappa
t)^2+\eta^2t^2]^{1/2}(1+\gamma+r^2-2r\kappa t)^{\ell/2}}.$$

Estimate of $J$ we start with the observation
$$\tau:=\min_{t\in [-1,1]}[(r-\kappa t)^2+\eta^2
t^2]=\min\{r^2\eta^2/\gamma,\, (r-\kappa)^2+\eta^2 \}.$$
Let
$\tau=r^2\eta^2/\gamma$, which is always the case if $r$ is
sufficiently small. In the case when $\tau=(r-\kappa)^2+\eta^2$ 
the proof is considerably
simpler and is left for the reader. If 
$\tau=r^2\eta^2/\gamma$, then
$$J\leq 2\pi \gamma^{1/2}\eta^{-1}\int_0^\infty dr\int_{-1}^{1}
dt[1+\gamma+r^2-2\kappa rt]^{-\ell/2}.$$

Integrating over $t$ yields
$$J\leq 2\pi \gamma^{1/2}\eta^{-1} [(\ell
-2)\kappa]^{-1}\mathcal{J},$$
where
$$\mathcal{J}:=\int_0^\infty
drr^{-1}[(1+\gamma+r^2-2\kappa r)^{-b}-(1+\gamma+r^2+2\kappa r)^{-b}],$$
and $b:=\ell/2 -1.$ 

Since $\eta=O(\ln \kappa)$, one has $\frac
{\eta}{\kappa}=o(1)$ as $\kappa\to \infty.$ Therefore,
$$\gamma^{1/2}\eta^{-1}\kappa^{-1}=O(\eta^{-1})\qquad as \quad \kappa\to
\infty.$$
Since $\ell>3$, one has $b>\frac 1 2$, and, as we prove below, 
\be\label{eq43}\mathcal{J}=o(\frac 1 \kappa)
\quad as \quad \kappa\to\infty.\ee
This relation implies the desired inequality: 
\be\label{eq44} J\leq o(\frac 1 \kappa)
\qquad as \quad \kappa\to \infty. \ee
Let us derive relation  \eqref{eq43}. One has
$$\mathcal{J}=\int_0^1+\int_1^\infty:=J_1+J_2,$$ $$J_1\leq \int_0^1
drr^{-1}\frac {(w^2+2r\kappa+r^2)^b-(w^2-2r\kappa+r^2)^b}
{(w^2+2r\kappa+r^2)^b(w^2-2r\kappa+r^2)^b},$$ where
$$w^2:=1+\gamma=1+\eta^2+\kappa^2.$$

Furthermore,
$$(w^2+2r\kappa+r^2)^b-(w^2-2r\kappa+r^2)^b\leq \frac {4br\kappa}
{(w^2-2r\kappa+r^2)^{1-b}}.$$
Thus,
$$J_1\leq 4b\kappa\int_0^1dr \frac 1
{(w^2+2r\kappa+r^2)^b(w^2-2r\kappa+r^2)}.$$
This implies the following estimate
$$J_1\leq O(\kappa/w^{2+2b})\leq O(\kappa^{-(1+2b)}),$$
because $w=\kappa [1+o(1)]$ as $\kappa \to \infty$.
Furthermore,
$$J_2\leq\int_1^\infty dr r^{-1}[(1+\eta^2+(r-\kappa)^2)^{-b}-
(1+\eta^2+(r+\kappa)^2)^{-b}]:=J_{21}-J_{22}.$$
One has $J_{22}\leq J_{21}$. 

Let us estimate $J_{21}$.
One obtains 
$$J_{21}=\int_1^{\kappa/2}+\int_{\kappa/2}^\infty:=j_1+j_2,$$ 
and
$$j_1\leq \frac 1 {[W^2+\frac {\kappa^2}4]^b} \ln \kappa=o(\frac 1
\kappa),\qquad W^2:=1+\eta^2,\qquad b>\frac 1 2.$$
Furthermore
$$j_2\leq \frac 2 \kappa \int_{\kappa/2}^\infty \frac
{dr}{[W^2+(r-\kappa)^2]^b}\leq \frac 2 \kappa \int_{-\infty}^\infty
\frac{dy}{[W^2+y^2]^b}=o(\frac 1 \kappa).$$
Thus, if $b>\frac 1 2$, then $J_2=o(\frac 1 \kappa)$ and
$\mathcal{J}=J_1+J_2=o(\frac 1 \kappa)$. Thus, relation \eqref{eq43} is 
proved. 

Relation \eqref{eq43} yields the desired
estimate
$$J=o(\frac 1 \kappa).$$
Thus, both estimates \eqref{eq43} and \eqref{eq44}
are proved. 

Note that the desired relation  $\mathcal{J}=o(\frac 1 \kappa)$
could have been obtained even by replacing $W^2$ by the smaller quantity 
$1$ in the above argument.  

Estimate \eqref{e42} implies 
\be\label{e45} I\leq 
ce^{a|\eta|}o\left(\frac{1}{\sqrt{\kappa^2+\eta^2}}\right),\quad
\kappa\to \infty,\quad \eta=a^{-1}\ln \kappa +O(1). \ee 
The quantity $\eta=\eta(k)=a^{-1}\ln \kappa +O(1)$ was chosen so that 
if  $\kappa \to \infty$, then  the quantity
$\frac{e^{|\eta|a}}{\sqrt{\kappa^2+\eta^2}}$
remains bounded as $\kappa\to \infty$. Therefore esimate \eqref{e45}
implies 
\be\label{e46}
\lim_{\kappa\to \infty,\eta=a^{-1}\ln \kappa+O(1)}I=0. \ee 
Consequently,  estimate \eqref{e40} holds.\\
\thmref{thm1} is proved. \hfill $\Box$

$$ $$
{\bf  APPENDIX}

{\it  1. Estimate of the norm of the operator $T^2$.}

Let \be\label{A1} Tf:=\int_{B_a}G(x-y,\kappa+i\eta)q(y)f(y)dy. \ee
Assume $q\in H_0^\ell(B_a)$, $\ell>2$, $f\in C(B_a)$. Our goal is to
prove that equation \eqref{e12} can be solved by iterations
for all sufficiently large $\kappa$. 

Consider $T$ as an
operator in $C(B_a)$. One has: 
\be\label{A2}\begin{split}
T^2f&=\int_{B_a}dzG(x-z,\kappa+i\eta)q(z)\int_{B_a}G(z-y,\kappa+i\eta)
q(y)f(y)dy\\
&=\int_{B_a}dyf(y)q(y)\int_{B_a}dzq(z)G(x-z,\kappa+i\eta)G(z-y,\kappa+i\eta).
\end{split}\ee
Let us estimate the integral 
\be\label{A3}\begin{split}
I(x,y):&=\int_{B_a}G(x-z,\kappa+i\eta)G(z-y,\kappa+i\eta)q(z)dz\\
&=\int_{B_a}\frac{e^{i(\kappa+i\eta)[|x-z|-\beta\cdot(x-z)+|z-y|-
\beta\cdot(z-y)]}}{16\pi^2|x-z||z-y|}q(z)dz\\
&=\frac{1}{16\pi^2}\int_{B_a}\frac{e^{i(\kappa+i\eta)[|x-z|+|z-y|-
\beta\cdot(x-y)]}}{|x-z||z-y|}q(z)dz\\
&:=\frac{e^{-i(\kappa+i\eta)\beta\cdot(x-y)}}{16\pi^2}I_1(x,y).\\
\end{split}\ee

Let us use the following coordinates (see \cite{R190}, p.391):
\be\label{A4} z_1=\ell st+\frac{x_1+y_1}{2},\quad
z_2=\ell\sqrt{(s^2-1)(1-t^2)}\cos\psi +\frac{x_2+y_2}{2}, \ee
\be\label{A5} z_3=\ell\sqrt{(s^2-1)(1-t^2)}\sin\psi
+\frac{x_3+y_3}{2}. \ee

The Jacobian $J$ of the ransformation $(z_1,z_2,z_3)\to
(\ell,t,\psi)$ is \be\label{A6} J=\ell^3(s^2-t^2),\ee where
\be\label{A7} \ell=\frac{|x-y|}{2},\quad |x-z|+|z-y|=2\ell s,\quad
|x-z|-|z-y|=2\ell t, \ee \be\label{A8}
|x-z||z-y|=4\ell^2(s^2-t^2),\quad 0\leq \psi<2\pi,\quad t\in[-1,1],\
s\in[1,\infty). \ee 
One has 
\be\label{A9} I_1=\ell\int_a^\infty
e^{2i(\kappa+i\eta)\ell s}Q(s)ds, 
\ee 
where 
\be\label{A10}
Q(s):=Q(s,\ell,\frac{x+y}{2})=\int_0^{2\pi}d\psi\int_{-1}^1dt
q(z(s,t,\psi;\ell,\frac{x+y}{2})), 
\ee 
and the function $Q(s)\in H_0^2(\R^3)$ for any fixed $x,y$. 
Therefore, an 
integration by parts in \eqref{A9} yields the following estimate: 
\be\label{A11}
|I_1|=O\left(\frac{1}{|\kappa+i\eta|}\right),\quad |\kappa+i\eta|\to 
\infty.
\ee 
From \eqref{A2}, \eqref{A3} and \eqref{A11} one gets:
\be\label{A12} \|T^2\|=O\left(\frac{1}{\sqrt{\gamma}}\right),\qquad
\gamma:=\kappa^2+\eta^2\to \infty. 
\ee 
Therefore, integral equation
\eqref{e12}, with $k$ replaced by $\frac{\kappa+i\eta}{2}$, can be
solved by iterations if $\gamma$ is sufficiently large and $\eta\geq
0$. Consequently, integral equation \eqref{e13} can be solved by
iterations. Thus, estimate \eqref{e40} holds if such an estimate
holds for the free term in equation \eqref{e13}, that is, for the function
$\frac{\tilde{q}}{\xi^2-(\kappa+i\eta)\beta\cdot \xi}$, namely,  if
estimate \eqref{e46} holds.
$$ $$
{\it 2. Proof of Lemma 2.1.}

Let $L_jG_j:=[\nabla^2 +k^2-q_j(x)]G_j(x,y,k)=-\delta(x-y)$ in $\R^3$, 
$j=1,2$. Applying Green's formula one gets
\be\label{A13} 
G_1(x,y,k)-G_2(x,y,k)=\int_{B_a}[q_2(z)-q_1(z)]G_1(x,z,k)G_2(z,y.k)dz.
\ee
In \cite{R190}, p. 46, the following formula is proved:
\be\label{A14} G_j(x,y,k)=\frac {e^{ik|y|}}{4\pi |y|}u_j(x,\alpha,k) 
+o(\frac 1 {|y|}),\qquad |y|\to \infty, \alpha:= -\frac{y}{|y|},
\ee  
where $u_j(x,\alpha,k)$ is the scattering solution, $j=1,2$.
Applying formula \eqref{A14} to \eqref{A13}, one obtains
\be\label{A15} 
u_1(x,\alpha,k)-u_2(x,\alpha,k)=\int_{B_a}[q_2(z)-q_1(z)]G_1(x,z,k)
u_2(z,\alpha,k)dz
\ee 
using the definition \eqref{e2} of the scattering amplitude $A(\beta, 
\alpha, k)$,
one derives from \eqref{A15} the relation
\be\label{A16}
4\pi [A_1(\beta, \alpha,k)-A_2(\beta, \alpha,k)]=\int_{B_a}[q_2(z)-q_1(z)]
u_1(z,-\beta,k)u_2(z,\alpha,k)dz.
\ee
This formula is equivalent to \eqref{e14} because of the
well-known reciprocity relation $A(\beta, \alpha, k)=A(-\alpha, 
-\beta,k)$.

Lemma 2.1 is proved.\hfill $\Box$

\end{document}